\documentclass[]{revtex4-2}

\usepackage{amsmath,amsthm,amssymb}
\usepackage{color}
\usepackage{natbib}
\usepackage[hidelinks]{hyperref}
\usepackage{cleveref}
\usepackage{graphicx}

\theoremstyle{definition}
\newtheorem{definition}{Definition}

\theoremstyle{plain}
\newtheorem{lemma}{Lemma}
\newtheorem{theorem}{Theorem}
\newtheorem{proposition}{Proposition}

\begin{document}
\title{The Lorenz system as a gradient-like system}
\author{Jeremy P. Parker}
\email{jparker002@dundee.ac.uk}
\address{
Division of Mathematics, University of Dundee, Dundee DD1 4HN, United Kingdom
}

\date{\today}
\begin{abstract}
    We formulate, for continuous-time dynamical systems, a sufficient condition to be a gradient-like system, i.e. that all bounded trajectories approach stationary points and therefore that periodic orbits,  chaotic attractors, etc. do not exist. This condition is based upon the existence of an auxiliary function defined over the state space of the system, in a way analogous to a Lyapunov function for the stability of an equilibrium. For polynomial systems, Lyapunov functions can be found computationally by using sum-of-squares optimisation. We demonstrate this method by finding such an auxiliary function for the Lorenz system. We are able to show that the system is gradient-like for $0\leq\rho\leq12$ when $\sigma=10$ and $\beta=8/3$, significantly extending previous results. The results are rigorously validated by a novel procedure: First, an approximate numerical solution is found using finite-precision floating-point sum-of-squares optimisation. We then prove that there exists an exact solution close to this using interval arithmetic.
\end{abstract}

\maketitle

\section{Introduction}

Nonlinear dynamical systems exhibit incredibly rich behaviour. 
A simple set of precisely deterministic differential equations can give rise to completely unpredictable chaos, while in other cases similar equations give no interesting dynamics at all, with all trajectories converging on the same equilibrium. One approach to understanding such systems is to find topologically simple trajectories, such as stationary points and periodic orbits, with the idea that more complex trajectories must approximately follow these when close. Even these `simple' trajectories, however, can be very difficult to detect, and the fact that they have not been found should not be taken as evidence that they do not exist.

Finding stationary points in systems of ordinary differential equations is straightforward, and for the polynomial systems studied in this work, this reduces to the well-studied problem of finding roots of polynomials. Often, one can be certain that they have found all the stationary points of a system.
Finding periodic orbits, however, is a much harder problem. Indeed, the famous 16th problem of David Hilbert \citep{hilbert}, unsolved for over a century, is concerned with understanding periodic orbits in two-dimensional polynomial systems. In three-dimensions, the problem is even more acute, since periodic orbits can form together in a chaotic attractor. The existence or otherwise of chaos within a system is a difficult problem with practical applications to everything from engineering design to planetary sciences.

Some periodic orbits are born from Hopf bifurcations of stationary points, and these are relatively easy to find. Others, however, appear in so-called global bifurcations. If dynamically unstable, these are can be incredible difficult to detect without a very close guess. One example of this is in the celebrated
 Lorenz \citep{lorenz1963deterministic} equations:
\begin{align}
\label{eq:lorenz}
\begin{split}
    \frac{\mathrm{d}x_1}{\mathrm{d}t} &= \sigma(x_2-x_1),\\
    \frac{\mathrm{d}x_2}{\mathrm{d}t} &= x_1(\rho-x_3)-x_2,\\
    \frac{\mathrm{d}x_3}{\mathrm{d}t} &= x_1 x_2 - \beta x_3.
    \end{split}
\end{align}
These were originally derived as a model for Rayleigh-Bénard convection, and the parameter $\rho$ represents a rescaled Rayleigh number. Like in the fluid dynamical problem, complexity increases as $\rho$ increases (for moderate values), initially having a single globally stable stationary point and eventually exhibiting chaos.

The Lorenz system is most often studied at parameter values $\beta=8/3$, $\sigma=10$, and $\rho>0$.
For $0<\rho<1$, there is a single globally attracting equilibrium at the origin. At $\rho=1$ this undergoes a supercritical pitchfork bifurcation to produce two initially stable equilibria. These remain stable until a subcritical Hopf bifurcation at $\rho=\sigma(\sigma+\beta+3)/(\sigma-\beta-1)\approx 24.7$ \citep{sparrow2012lorenz}.
When $\rho=28$, this system exhibits the famous `butterfly wing' chaotic attractor, which persists as a strange invariant set, along with its infinity of unstable periodic orbits, as low as $\rho=13.927$ \citep{sparrow2012lorenz, chen1996lorenz}. At this value, the heteroclinic connections between the equilibria undergo a global bifurcation, and an `explosion' of periodic orbits occurs.

No periodic orbits have been found for $0<\rho<13.926$, but results guaranteeing their nonexistence are limited. Most of the research has been concerned with guaranteeing that chaos does not exist, which is a stronger condition, but has more practical relevance, and can be investigated through Lyapunov exponents. \citet{martini2022ruling}\citep{martini2023bounding} use a method of \citet{angeli2021robust} to show that positive Lyapunov exponents on the attractor, i.e. chaos, cannot exist for $\rho\leq 2\sqrt{\beta-1}(\beta+\sigma)(\beta+1)/\sigma\beta \approx 4.5$.
\citet{leonov2016lyapunov} proved that if $\rho < (\beta+1)(\beta+\sigma)/\sigma \approx 4.6$, all bounded trajectories converge to one of the equilibria. Above this, they proved an upper bound on the dimension of the global attractor which is strictly greater than 2, and thus does not prohibit chaos.
In this paper, we will show that periodic orbits and chaos cannot exist for $0\leq\rho\leq 12$, a significantly higher maximum value than those mentioned above. In fact, we will prove global stability. This term, coined in this context by \citet{kuznetsov2020birth}, means that all trajectories in the system converge to one of the system's equilibria.

The paper is laid out as follows: In section \ref{sec:theorem}, we present a general Lyapunov function method, and specialise this to the Lorenz system. Section \ref{sec:sos} presents the sum-of-squares relaxation of this method which is computationally tractable. In section \ref{sec:validation}, we derive and apply a new method for the rigorous validation of these results using interval arithmetic. Section \ref{sec:conc} summarises the work and discusses potential extensions to the methods described. An appendix presents an analogous method for maps.

\section{Mathematical results}
\label{sec:theorem}
Let $\frac{\mathrm{d}x}{\mathrm{d}t} = f(x)$ be a dynamical system for a continuously differentiable function $f:\mathbb{R}^n\to\mathbb{R}^n$, such that the corresponding flow $\varphi:\mathbb{R}^n\times\mathbb{R}\to\mathbb{R}^n$ with $\frac{\mathrm{d}}{\mathrm{d}t}\varphi(x,t)=f(x)$ exists and is continuously differentiable for all $x\in\mathbb{R}_{>0}^n$ and $t\in\mathbb{R}$. Recall the following definition \citep{glendinning1994stability}:
\begin{definition}
The \textit{$\omega$-limit set} of a point $x$ with respect to a flow $\varphi$ is the subset of state space which is visited at arbitrary times,
\begin{equation}
    \omega(x,\varphi) = \{y\in\mathbb{R}^n \mid \exists (t_n)\text{ such that }t_n\to\infty,\,\varphi(x,t_n)\to y\text{ as }n\to\infty\}.
\end{equation}

A \textit{limit point} of the system is $y\in\mathbb{R}^n$ such that $y\in\omega(x,\varphi)$ for some $x$. The \textit{limit set} of the system is the collection of all such points.
\end{definition}

\begin{lemma}
\label{thm:main}
    Let $g:\mathbb{R}^n\to\mathbb{R}_{\geq0}$ be a continuous, nonnegative function defined over state space. Suppose there exists a continuously differentiable function $V:\mathbb{R}^n\to\mathbb{R}$ which satisfies
    \begin{equation}
        f(x)\cdot \nabla V(x) \geq g(x)
    \end{equation}
    for all $x\in \mathbb{R}^n$. Then $g(y)=0$ for all limit points $y$.
\end{lemma}
\begin{proof}
Suppose that $V:\mathrm{R}^n\to\mathbb{R}$ satisfies the condition of the theorem.

Choose any $x_0\in\mathbb{R}^n$ and $y\in \omega(x_0,\varphi)$. In particular, let $(t_n)$ be a sequence of times such that $\lim_{n\to\infty} \varphi(x_0,t_n) = y$. Then
\begin{align*}
    V(\varphi(x_0,t_n)) - V(\varphi(x_0,t_{n-1})) &= \int_{t_{n-1}}^{t_n} \frac{\mathrm{d}}{\mathrm{d}t} V(\varphi(x_0,t)) \mathrm{d} t \\ 
    &=\int_{t_{n-1}}^{t_n} f(\varphi(x_0,t)) \cdot \nabla V(\varphi(x_0,t))\mathrm{d} t \\
    &\geq \int_{t_{n-1}}^{t_n} g(\varphi(x_0,t)) \mathrm{d} t \\
    &\geq 0.
\end{align*}
But $V(\varphi(x_0,t_n)) - V(\varphi(x_0,t_{n-1})) \to 0$ as $n\to\infty$ by the continuity of $V$, so
\begin{equation*}
\int_{t_{n-1}}^{t_n} g(\varphi(x_0,t)) \mathrm{d} t\to 0 \text{ as } n\to\infty.
\end{equation*}

Since $g(x)\geq0$ for all $x$, we must therefore have $g(\varphi(x_0,t))\to 0$ as $t\to\infty$, and in particular, 
\begin{equation*}g(y) = \lim_{n\to\infty} g(\varphi(x_0,t_n)) = 0.\end{equation*}

\end{proof}

An immediate application of this result is the following:




\begin{theorem}
\label{thm:globalstability}
Suppose there exists a smooth function $V:\mathbb{R}^n\to\mathbb{R}$ which satisfies
    \begin{equation}
        f(x)\cdot \nabla V(x) \geq \|f(x)\|^2
    \end{equation}
    for all $x\in \mathbb{R}^n$. Then the limit set of the system $f$ is the set of stationary points.
\end{theorem}
\begin{proof}
For every limit point $x$, we know from \cref{thm:main} that $\|f(x)\|^2=0$ and thus $f(x)=0$, i.e. the point is a stationary point. Conversely, any stationary point is its own limit point.
\end{proof}

Notice that this implies $f\cdot\nabla V$ is strictly positive away from stationary points, so this condition is a special case of a strict Lyapunov function as discussed by \citet{barta2012every}, and they prove that this means the system is a gradient system when restricted to the subset of non-stationary points.
Such systems are known by various names \citep{kuznetsov2020birth}.
We prefer the term \textit{gradient-like} systems \citep{hale2004stability}, since an explicit gradient is not easily found. Nevertheless, the behaviour of such a system is particularly simple. Certainly no periodic orbits can exist, since any periodic point is contained in its own $\omega$-limit set, and these points must necessarily be non-stationary.
Similarly, quasiperiodic motion is forbidden in a gradient-like system, though this becomes more relevant in higher dimensions than 3.
There is no universally agreed definition of chaos, but the most widely cited is that of \citet{devaney2018introduction}, who takes as an axiom that periodic points are dense within a chaotic set. Therefore, gradient-like systems are non-chaotic. Note that, however, a gradient-like system does not imply that all limit points are \textit{isolated} stationary points.


\subsection{Application to the Lorenz system}
If there are no unbounded solutions to a gradient-like system, for example because of the existence of globally absorbing compact set (sometimes called \textit{dissipativity in the sense of Levinson}), we can further deduce that all trajectories are attracted to the stationary set. This property has been called \textit{global stability} \citep{kuznetsov2020birth}, and it is what we shall attempt to prove for the Lorenz system.
First recall that \citep{lorenz1963deterministic}:
\begin{theorem}
\label{thm:dissipative}
    The system \labelcref{eq:lorenz} is dissipative in the sense of Levinson.
\end{theorem}
In fact, various different simple, globally attracting sets can be found, the compact intersection of which can be used to estimate the shape of the attractor \citep{doering1995shape}.

For the remainder of this work -- though there is no reason to suppose these methods cannot be easily applied to other parameter values -- we shall concentrate on the classic case of $\sigma=10$ and $\beta=8/3$. This poses a slight problem for the numerical validation, since $8/3$ cannot be exactly represented using IEEE 754 floating point numbers. Therefore, we rescale \labelcref{eq:lorenz} so that $x_1' = x_1/M$, $x_2' = x_2/M$, $x_3' = x_3/M$ and $t' = t/3$, and also let $r = (\rho-1)/M$, so that (immediately dropping the primes),
\begin{align}
\label{eq:lorenzresc}
\begin{split}
    \frac{\mathrm{d}x_1}{\mathrm{d}t} &= 30x_2-30x_1,\\
    \frac{\mathrm{d}x_2}{\mathrm{d}t} &= 3Mx_1r- 3Mx_1x_3-3x_2+3x_1,\\
    \frac{\mathrm{d}x_3}{\mathrm{d}t} &= 3Mx_1 x_2 - 8x_3,
\end{split}\end{align}
where $M$ is an arbitrary integer factor chosen the rescale the numerics \citep{goluskin2018bounding,parker2021study}, which will be helpful later. 
The following result is an alternative to \cref{thm:globalstability} which will lead to a less computationally complex sum-of-squares program.
\begin{proposition}
\label{thm:lorenz}
Let $f:\mathbb{R}^3\to\mathbb{R}^3$ be defined by 
\begin{equation}
    f(x_1,x_2,x_3)=\left(30x_2-30x_1,\,3Mx_1r- 3Mx_1x_3-3x_2+3x_1,\,3Mx_1 x_2 - 8x_3\right).
\end{equation}    
If there exists $V:\mathbb{R}^3\to\mathbb{R}^3$ such that
    \begin{equation}
        f(x)\cdot\nabla V(x) \geq (x_2-x_1)^2
    \end{equation}
    then the system \labelcref{eq:lorenzresc} is globally stable.
\end{proposition}
\begin{proof}
The system is equivalent to \labelcref{eq:lorenz} and therefore dissipative by \cref{thm:dissipative}.

    Suppose such a $V$ exists. By \cref{thm:main}, we then have that $x_2-x_1=0$ for every limit point $x$.
    Then at a limit point
    \begin{equation*}
    0 = \frac{\mathrm{d}}{\mathrm{d}t}(x_2-x_1) = f_2(x)-f_1(x) = f_2(x) - 30(x_2-x_1) = f_2(x) = 3Mx_1r- 3Mx_1x_3-3x_2+3x_1.
    \end{equation*}
Therefore $x_1(r-x_3) = 0$.

If $x_1=x_2=0$, we have a one-dimensional system $\frac{\mathrm{d}x_3}{\mathrm{d}t}=-8 x_3$. Taking the antiderivative as a gradient, all one-dimensional systems are gradient systems and thus all bounded trajectories are attracted to stationary points.

Otherwise, $x_3=r$, and $0=\frac{\mathrm{d}x_3}{\mathrm{d}t}=x_1^2-8r$, so $x_1 = x_2 = \pm\sqrt{8r}$, which are the two non-trivial equilibria of the system when $r>0$. So all limit points are stationary points.

Furthermore, the system is dissipative in the sense of Levinson so it is globally stable.
\end{proof}

As stated, \cref{thm:lorenz} could only prove global stability for a particular value of the parameter $r$. We would like to numerically prove global stability for a range of $r$, which suggests treating it as a fourth variable in the system. This leads to the following result:
\begin{proposition}
\label{thm:lorenzrange}
Let $f:\mathbb{R}^3\times\mathbb{R}\to\mathbb{R}^3$ be defined by 
\begin{equation}
\label{eq:systemtosolve}
    f(x_1,x_2,x_3;r)=\left(30x_2-30x_1,\,3Mx_1r- 3Mx_1x_3-3x_2+3x_1,\,3Mx_1 x_2 - 8x_3\right).
\end{equation}
If there exists $V:\mathbb{R}^3\times\mathbb{R}\to\mathbb{R}$ and $s:\mathbb{R}^3\times\mathbb{R}\to\mathbb{R}_{\geq0}$ such that
    \begin{equation}
        f(x;r)\cdot\nabla V(x;r) \geq (x_2-x_1)^2 + (Mr+1-\rho_{min})(\rho_{max}-Mr-1) s(x;r)
    \end{equation}
then the system \labelcref{eq:lorenz} with $\beta=8/3$ and $\sigma=10$ is globally stable for $\rho_{min} \leq \rho \leq \rho_{max}$.
\end{proposition}
\begin{proof}
System \labelcref{eq:lorenz} is equivalent to system \labelcref{eq:lorenzresc} with $r=(\rho-1)/M$. If $\rho_{min} \leq \rho \leq \rho_{max}$, then 
\begin{equation*}
    (Mr+1-\rho_{min})(\rho_{max}-Mr-1)\geq 0,
\end{equation*}
and so
\begin{equation*}
        f(x;r)\cdot\nabla V(x;r) \geq (x_2-x_1)^2
\end{equation*}
since $s\geq 0$. Therefore, by \cref{thm:lorenz}, system \labelcref{eq:lorenzresc} is globally stable and hence also system \labelcref{eq:lorenz}.
\end{proof}

Introducing the second auxiliary function $s$ is an example of the generalised S-procedure, which is frequently used in sum-of-squares optimization \citep{tanthesis,fantuzzi2016bounds}.
Notice that with $r$ as a variable, the system \labelcref{eq:systemtosolve} is a polynomial system with integer coefficients which can be exactly represented by a computer (assuming $M$ is not unreasonably large).

\section{Sum-of-squares relaxation}
\label{sec:sos}
Finding Lyapunov functions for stability problems can be difficult, but one systematic approach is through sum-of-squares optimization \citep{papachristodoulou2002construction}. If the system and the auxiliary functions are both restricted to be polynomials of relatively low degree, the problem often becomes computationally tractable. Verifying that a polynomial is non-negative is NP-hard \citep{murty1985some}, but verifying that it is a sum of squares of polynomials, a stronger but sufficient condition, is possible through semi-definite programming \citep{parrilo2003semidefinite}.

Let $P_n$ be the space of polynomials in $n$ real variables (of any degree).
For a polynomial $f \in P_n$, we say that $f$ is a sum-of-squares, written $f\in \Sigma_n$, if there exist polynomials $p_1,...,p_m\in P_n$ such that
\begin{equation*}
    f = \sum_{k=1}^m p_k^2.
\end{equation*}
Clearly $f\in\Sigma_n$ is a non-negative polynomial, though in general the converse is not true.
Let $\mathbb{S}_m$ be the space of $m\times m$ symmetric matrices which are positive semi-definite.
An important result \citep{choi1995sums} is that $f\in\Sigma_n$ if and only if there exists a (non-unique) $Q\in \mathbb{S}_m$ and a choice of `basis' functions $b_1,\dots,b_m \in P_n$ such that
$f = \mathbf{b}^T Q \mathbf{b}$. The matrix $Q$ is called the Gram matrix.

\begin{figure}
    \centering
    \includegraphics[width=0.8\textwidth]{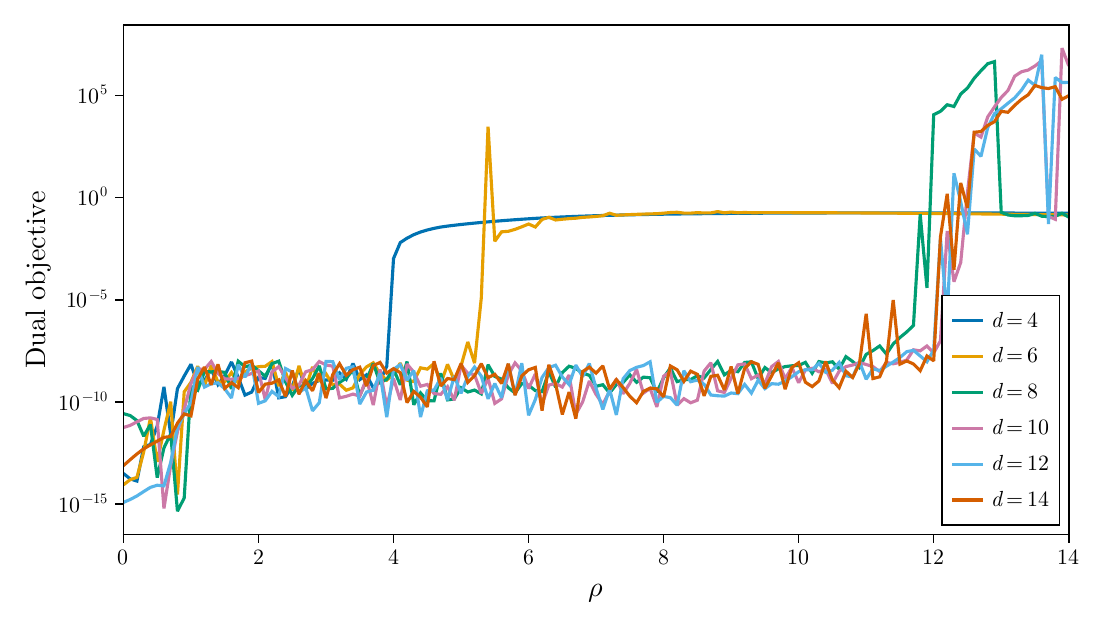}
    \caption{Dual objective value reported by the SDP solver for a na\"ive implementation of \cref{thm:lorenz} (with $M=16$). A value close to zero indicates that the SDP is likely to be feasible, i.e. that the Lorenz system is indeed gradient-like at that value of $\rho$.}
    \label{fig:lorenz}
\end{figure}

A straightforward relaxation of \cref{thm:lorenz} gives a sum-of-squares program, which itself is equivalent to a semi-definite program (SDP) which can be solved numerically with a number of existing software packages. We employ the Julia package SumOfSquares.jl \citep{legat2017sos,weisser2019polynomial}, on top of the MOSEK SDP optimizer \citep{MOSEK}. \Cref{fig:lorenz} shows the dual objective value, effectively a residual, returned by MOSEK for this SDP, with different choices of $\rho$ and polynomial degree. A dual objective value close to zero indicates that the problem is likely to be feasible. For $d=4$ and $d=6$, a significant jump in the value is observed around $\rho=4$ and $\rho=5.5$ respectively, indicating where a Lyapunov function of that degree is no longer sufficient. For $d=8$, a clear jump is observed around $\rho=12$, and increasing the degree further does not improve this significantly.
These results suggest that the Lorenz system is gradient-like up to at least $\rho=12$.
However, these numerical results are only indicative and are not proofs, and furthermore they are valid only at a a given fixed value of $\rho$ and are not useful for the general results we want to prove.
Therefore we relax \cref{thm:lorenzrange} to the following:
\begin{proposition}
\label{thm:sos}
Let $f_1,f_2,f_3\in P_4$ be defined by
\begin{equation}
\label{eq:thmsys}
    f(x_1,x_2,x_3;r)=\left(30x_2-30x_1,\,3Mx_1r- 3Mx_1x_3-3x_2+3x_1,\,3Mx_1 x_2 - 8x_3\right).
\end{equation}
If there exist polynomials $V,b_1,\dots,b_{m_2}\in P_4$ and $s,a_1,\dots,a_{m_1}\in P_3$ and matrices $Q\in \mathbb{S}_{m_1}$ and $R\in \mathbb{S}_{m_2}$ such that
\begin{align}
\label{eq:SDP}
        s &= \mathbf{a}^T Q \mathbf{a}, \\
        f_1\partial_1 V + f_2\partial_2 V + f_3\partial_3 V - (x_2-x_1)^2 - M^2(r-r_{min})(r_{max}-r) s &= \mathbf{b}^T R \mathbf{b},
    \end{align}
then the system \labelcref{eq:lorenz} with $\beta=8/3$ and $\sigma=10$ is globally stable for $\rho_{min} \leq \rho \leq \rho_{max}$, where $\rho_{min} = Mr_{min}+1$ and $\rho_{max} = Mr_{max}+1$.
\end{proposition}

For fixed bases $\mathbf{a}$ and $\mathbf{b}$, this is a semi-definite program, which can be solved numerically for the coefficients of $Q$ and $R$ and the coefficients of the polynomials $V$ and $s$ with respect to the monomial basis. Nevertheless, such optimisations are slow, imprecise and difficult to validate if applied naively, and can be improved by a careful choice of bases.

The first observation is that \labelcref{eq:lorenzresc} is invariant under the symmetry transformation $(x_1,x_2,x_3,r)\mapsto(-x_1,-x_2,x_3,r)$, and so if this is true for $V$ and $s$, the left-hand sides of \labelcref{eq:SDP} is also invariant. In order for the right-hand sides to be invariant under this symmetry, the Gram matrices are necessarily block diagonal, and we can subdivide the bases and matrices such that
\begin{align*}
\mathbf{a}^T Q \mathbf{a} = \mathbf{a_e}^T Q_e \mathbf{a_e} + \mathbf{a_o}^T Q_o \mathbf{a_o},\\
\mathbf{b}^T R \mathbf{b} = \mathbf{b_e}^T R_e \mathbf{b_e} + \mathbf{b_o}^T R_o \mathbf{b_o},
\end{align*}
where $\mathbf{a_e}$ and $\mathbf{b_e}$ are invariant under the symmetry transformation and $\mathbf{a_o}$ and $\mathbf{b_o}$ change sign. Since $\mathbf{a_e}$ and $\mathbf{a_o}$ are roughly half the size of $\mathbf{a}$, we have reduced the number of unknown matrix entries by a factor of four, which results in a very significant speedup of both SDP and validation.

Secondly, notice that at a stationary point of the system, $f_1\partial_1 V + f_2\partial_2 V + f_3\partial_3 V - (x_2-x_1)^2 = 0$. Therefore, for the equalities and definiteness to hold, we must have that $s=0$ and also $\mathbf{a}^T Q \mathbf{a} = \mathbf{b}^T R \mathbf{b} = 0$ at stationary points. This suggests choosing $\mathbf{a_e},\mathbf{b_e},\mathbf{a_o},\mathbf{b_o}$ which themselves vanish at the stationary points, as was done (at much lower degree) in the analytic results of \citet{goluskin2018bounding}.

One further optimisation would be to choose a specific form of the polynomials $V$ and $s$, following \citet{goluskin2018bounding}. This was found to make little difference to performance, so for clarity of code we simply used a set monomials invariant under the symmetry transformations. The Lyapunov function $V$ was taken to be a polynomial of degree at most $d_x$ in $x_1$, $x_2$ and $x_3$ combined, and of maximum degree $d_r$ in $r$ (so that terms like $x_1^{d_x} r$ are allowed, but not $x_1^{d_x} x_2$). The function $s$ was taken to be a polynomial of maximum degree $d_x$ in $x$ and maximum degree $d_r-1$ in $r$.

For the validation procedure described in section \ref{sec:validation} to be successful, it will be necessary that the Gram matrices are strictly positive definite, rather than simply semi-definite. In other words, we need non-singular Gram matrices. We can attempt to enforce this by a parsimonious choice of basis, so that no unnecessary zero eigenvectors are included. Successful bases were found through trial-and-error, using the principles above, and these are listed in \cref{tab:bases} for $d_x=8$. Notice that all of the polynomials listed vanish at each of the system's three equilibria.
Partial results, with lower $\rho_{max}$ and lower $d_x$ were also successful; the bases for these are included in the source code.

\begin{table}[h]
    \centering
    \begin{tabular}{c|c}
        $\mathbf{a}_e$ & $(x_2-x_1)(x_1,x_2,x_1x_3,x_2x_3,x_1x_3^2,x_2x_3^2), (3Mx_1x_2-8x_3)(1,x_3,x_1^2,x_1x_2,x_2^2,x_3^2)$ \\ \hline
        $\mathbf{a}_o$ & $(x_2-x_1)(1,x_3,x_1^2,x_1x_2,x_2^2,x_3^2,x_1^2x_3,x_1x_2x_3,x_2^2x_3),(3Mx_1x_2-8x_3)(x_1,x_2,x_1x_3,x_2x_3)$\\\hline
         & $(x_2-x_1)(x_1,x_2,x_1x_3,x_2x_3,x_1x_3^2,x_2x_3^2,rx_1,rx_2,rx_1x_3,x_2x_3,rx_1x_3^2,rx_2x_3^2),$ \\
        $\mathbf{b}_e$ & $(r-x_3)(x_3,x_1^2,x_1x_2,x_2^2,x_3^2,x_1^2x_3,x_1x_2x_3,x_2^2x_3)$, \\
         & $(3Mrx_2^2-8x_3^2)(1,x_3,x_1^2,x_1x_2,x_2^2,x_3^2)$ \\\hline
          & $(x_2-x_1)(1,x_3,x_1^2,x_1x_2,x_2^2,x_3^2,x_1^2x_3,x_1x_2x_3,x_2^2x_3,r,rx_3,rx_1^2,rx_1x_2,rx_2^2,rx_3^2,rx_1^2x_3,rx_1x_2x_3)$,\\
        $\mathbf{b}_o$ & $(r-x_3)(x_1,x_2,x_1x_3,x_2x_3,x_1x_3^2,x_2x_3^2)$, \\
        & $(3Mrx_2^2-8x_3^2)(x_1,x_2,x_1x_3,x_2x_3)$
    \end{tabular}
    \caption{Bases of polynomials used for Gram matrices, for Lyapunov function of degree $(d_x,d_r) = (8,1)$. Notice that $\mathbf{a}_e$ and $\mathbf{a}_o$ have maximum degree 4 in $x$ and do not depend on $r$, and $\mathbf{b}_e$ and $\mathbf{b}_o$ have maximum degree 4 in $x$ and 1 in $r$.}
    \label{tab:bases}
\end{table}

\section{Validation with interval arithmetic}
\label{sec:validation}
The rigorous validation of numerical results of dynamical systems theory through the use of interval arithmetic have been some of the major successes of the field in the last few decades \citep{tucker1999lorenz,figueras2017numerical,wilczak2020geometric}.
However, once it has been demonstrated to be possible, it is not usually considered necessary to validate every single result. For example, if one has converged a periodic orbit to machine precision for a given system, and the location and period of this periodic orbit appears to be robust to variations in the timestep or discretisation of the numerical solution, then there is every reason to expect that this is a genuine solution to the underlying differential equations without the laborious process of validation through interval arithmetic. 
Unfortunately this is not the case for sum-of-squares optimisation. The underlying semidefinite programs are known to be imprecise and very sensitive to numerical artifacts. Furthermore, sum-of-squares programs are observed to generate particularly poorly-conditioned SDPs, likely due to the fact that very high degree exponents are applied to finite precision floating number. Solvers such as MOSEK can report success and return solutions with small residuals, even when the problem is in fact infeasible, or report failure but return plausible-seeming results. Therefore it is imperative that we can validate our numerics.

Several different procedures for validating sum-of-squares relaxations of this type are found in the literature. The simplest uses the numerical result to write down an analytic Lyapunov function whose validity is confirmed explicitly \citep{goluskin2018bounding}. This is practical only in relatively simple cases, although specially designed computer-algebra packages make this easier \citep{henrion2019spectra}.
A different method is to project the problem onto the space of sums-of-squares of polynomials with rational coefficients \citep{peyrl2008computing}, which can be computed exactly and for which a package is also available \citep{cifuentes2020}.
A third option is to use interval arithmetic to validate an approximate numerical result. For optimisation problems, one can find rigorous, but not sharp, upper- and lower-bounds on a quantity of interest by perturbing the result of an SDP solver \citep{jansson2008rigorous}. This has been applied to find rigorous bounds for dynamical systems \citep{goluskin2018bounding}. However, the method discussed in the present work is equivalent to finding upper- and lower-bounds of exactly zero, so this cannot be employed directly. 
We therefore derive a different, simpler method to validate an approximate result for a sum-of-squares feasibility problem with interval arithmetic, which formulates the problem as a linear system.

\subsection{A semi-definite program as a linear system}
Consider the problem of finding $s\in P_n$ and $Q\in\mathbb{S}_m$ such that \begin{equation}
\label{eq:example}
    s = \mathbf{a}^T Q \mathbf{a},
\end{equation}for some given set of polynomials $a_1,...,a_m\in P_n$. (With no other conditions this example is of course trivial.) Let $\mathbf{q}\in\mathbb{R}^{m(m+1)/2}$ be the vector of unique elements of the symmetric Gram matrix Q, and let $\mathbf{s}\in\mathbb{R}^N$ be the coefficients of $s$ with respect to the monomial basis. \Cref{eq:example} is linear in these coefficients, and so expanding it in monomials and comparing coefficients, we can formulate this as a system
\begin{equation}
    A \left(\begin{matrix}\mathbf{s} \\ \mathbf{q}\end{matrix}\right) = \mathbf{c}
\end{equation}
for some matrix of coefficients $A\in \mathbb{R}^{N \times (N+m(m+1)/2)}$ and vector of constant terms  $\mathbf{c}\in \mathbb{R}^N$. (In this example $\mathbf{c}=0$.)
Suppose that the rank of $A$ is $M$, and that we have some approximate solution $\mathbf{s}_0$, $\mathbf{q}_0$ such that $A \left(\begin{matrix}\mathbf{s_0}\cr \mathbf{q_0}\end{matrix}\right) - \mathbf{c} = \mathbf{\epsilon}$ is small.
Then we can project onto the solution space of the system using
\begin{equation}
    \left(\begin{matrix}\mathbf{s}\cr \mathbf{q} \end{matrix}\right) = A^+ c +(I-A^+ A) \left(\begin{matrix}\mathbf{s_0}\cr \mathbf{q_0} \end{matrix}\right) = \left(\begin{matrix} \mathbf{s_0}\cr \mathbf{q_0} \end{matrix}\right) - A^+ \mathbf{\epsilon}
\end{equation}
where $A^+$ is the Moore-Penrose right inverse, i.e. $A A^+ = I$.
Reassembling this into a polynomial $s$ with coefficients $\mathbf{s}$ and a Gram matrix $Q$ with coefficients $\mathbf{q}$ , we have a solution to \cref{eq:example}. There is no guarantee that the resulting $Q$ is positive semi-definite, but since the matrix $Q_0$ with coefficients $\mathbf{q_0}$ is positive semi-definite, there is reason to hope $Q$ will be too if the initial error $\mathbf{\epsilon}$ is sufficiently small. Our strategy will be to find a vector of intervals which is guaranteed to contain a solution to the linear system, and prove that all $\mathbf{q}$ in these intervals lead to a positive semi-definite Gram matrix, as summarised in the following:
\begin{proposition}
    \label{thm:linsos}
    Let $\mathbf{v}$ be the coefficients of the unknown polynomial functions $V(x,r)$ with respect to the basis of monomials up to degree $d_x$ in $x$ and $d_r$ in $r$. Let $\mathbf{q}_e$, $\mathbf{r}_e$, $\mathbf{q}_o$ and $\mathbf{r}_o$ be the unique elements of the unknown real symmetric matrices $Q_e$, $R_e$, $Q_o$ and $R_o$ respectively, and let $A$ and $\mathbf{c}$ be the (known) matrix and vector to represent the polynomial equation
    \begin{equation}
    \label{eq:fullpoly}
        f_1\partial_1 V + f_2\partial_2 V + f_3\partial_3 V - (x_2-x_1)^2 = M^2(r-r_{min})(r_{max}-r)\left(\mathbf{a_e}^T Q_e \mathbf{a_e} + \mathbf{a_o}^T Q_o \mathbf{a_o}\right) + \left(\mathbf{b_e}^T R_e \mathbf{b_e} + \mathbf{b_o}^T R_o \mathbf{b_o}\right)
    \end{equation}
    as the linear system
    \begin{equation}
    \label{eq:fulllin}
        A \left(\begin{matrix}
        \mathbf{v}\cr\mathbf{q}_e\cr\mathbf{r}_e\cr\mathbf{q}_o\cr\mathbf{r}_o
        \end{matrix}\right) = \mathbf{c}.
    \end{equation}

    If $A$ has full row-rank, and, for some vector $\mathbf{y}_0$, the solution
    \begin{equation}
    \left(\begin{matrix}
        \mathbf{v}\cr\mathbf{q}_e\cr\mathbf{r}_e\cr\mathbf{q}_o\cr\mathbf{r}_o
        \end{matrix}\right) = A^+ \mathbf{c} +(I-A^+ A) \mathbf{y_0}
\end{equation}
leads to symmetric matrices $Q_e$, $R_e$, $Q_o$ and $R_o$ which are positive semi-definite, then the system \labelcref{eq:lorenzresc} is globally stable for $r_{min}\leq r \leq r_{max}$.
\end{proposition}
Note that we have entirely eliminated the additional auxiliary function $s$ from this problem, it is used only to find an initial guess $\mathbf{y_0}$ through approximate sum-of-squares optimisation.

One key point here is that the matrix $A$ must be of full row-rank. It is actually the case that by na\"ively assembling $A$ by comparing the coefficients of each of the possible monomials, one arrives at a matrix which has degenerate rows. At a stationary point, where $f(x_1,x_2,x_3,r)=0$, every term in the polynomial expression \labelcref{eq:fullpoly} vanishes. This is by design, by construction of $\mathbf{a}_e$, $\mathbf{b}_e$, $\mathbf{a}_o$ and $\mathbf{b}_o$ as discussed in section \ref{sec:sos}. At a stationary point, either $x_1=x_2=x_3=0$ (which we can ignore by not considering monomials for which none of these appear), or $x_2=x_1$ and $r=x_3=3Mx_1^2/8$. Substituting these expressions into \cref{eq:fullpoly}, all terms vanish. However, subsituting these expressions in to the set of monomials which are used to generate \cref{eq:fulllin}, we arrive at a set of monomials in $x_1$ whose coefficients must vanish.
For example, the monomials $x_1^2$, $x_1x_2$, $x_2^2$, $x_3$ and $r$ all reduce to a monomial $x_1^2$ at a stationary point, and therefore the sum  (up to a constant factor of $3M/8$) of the rows in $A$ which correspond to these monomials must be zero.
This is straightforward to remedy: simply eliminate from \cref{eq:fulllin} one row for each power of $x_1$ in this reduced set of monomials -- since each row must be a linear combination of the others for that power of $x_1$.

\subsection{Interval arithmetic}
The basic idea of numerical interval arithmetic is, rather than using imprecise floating-point numbers directly for computations, to use intervals whose end-points are floating-point numbers but which certainly contain the exact result. For example, adding two intervals results in an interval which is typically slightly larger than the interval given by the approximate sum of the endpoints of the two intervals.
A comprehensive background on interval arithmetic is outside the scope of this work. The textbooks by \citet{jaulin2001applied} and \citet{tucker2011validated} give a good introduction to the field.

We have carefully constructed the problem to be validated in \cref{thm:linsos} so that the matrix $A$ and right-hand side $\mathbf{c}$ can be exactly expressed either as integers or as IEEE double-precision floating point numbers. However, constructing a right-inverse for $A$ with floating-point arithmetic will almost certainly lead to an imprecise result. For that reason, this step is performed with interval arithmetic, leading to an interval matrix which is guaranteed to contain $A^+$. 

Let $\mathbb{IR}$ be the space of closed but not necessarily bounded intervals on the real line, specified by their endpoints. Addition and subtraction are defined in the obvious ways, by addition and subtraction of the endpoints, and multiplication is defined to properly account for signs when intervals are negative or contain zero \citep{tucker2011validated}. These operations are implemented numerically using the Julia package IntervalArithmetic.jl \citep{IntervalArithmetic.jl}, so that the result of any calculation is guaranteed to contain the exact result. We denote by square brackets $[\cdot]$ an interval vector in $\mathbb{IR}^n$ (or interval matrix in $\mathbb{IR}^{n\times m}$) which contains the true value of a quantity in $\mathbb{R}^n$ (resp. $\mathbb{R}^{n\times m}$).

The right-inverse for a real matrix is usually expressed as
\begin{equation*}
    A^+ = A^T (A A^T)^{-1},
\end{equation*}
where $A^T$ is the usual transpose. In fact, it is easier to calculate as
\begin{equation*}
    A^+ = \left(\left(A A^T\right)^{-1} A\right)^T.
\end{equation*}
First we calculate an interval matrix $[A A^T]$ strictly enclosing the product $A A^T$ (both of which are known exactly), and we then use interval Gaussian elimination \citep{jaulin2001applied} with $[A A^T]$ and $A$ to find an interval matrix strictly containing $\left(A^+\right)^T$, which can be exactly transposed to give $[A^+]$. One important detail, as mentioned in the previous section, is that for this to be possible, the rectangular matrix $A$ must have full row-rank, otherwise the Gaussian elimination will give an error or result in infinite intervals and the validation will ultimately fail.

The final step of our validation will be to check that the Gram interval matrices are positive definite.
There are a number of different equivalent conditions for a real symmetric matrix to be positive definite.
The easiest to implement with interval arithmetic is to check that the determinants of all the principal minors, i.e. the upper-left $1\times1$ submatrix, the upper-left $2\times2$ submatrix etc., are positive. Since determinants are calculated with addition, subtraction and multiplication, no new code is required to compute these for interval matrices.
Notice that for a semi-definite matrix which is not strictly definite, the determinant is zero. When the determinant is calculated using interval arithmetic, the result will be an interval containing zero which is not strictly non-negative, and thus our validation will be unsuccessful. This is the reason that we must restrict to strictly positive definite matrices.

The sum-of-squares optimisation was performed using 64-bit `double' precision floating point numbers, to provide an initial guess. However, for validating these results, using 64-bit endpoints to represent the intervals was insufficient and resulted in determinants which were slightly negative. 128-bit floating point numbers were sufficient for almost all our results, but for higher values of $\rho$ it was necessary to use 256 bits. This is straightforward, if rather slow, with Julia.

\subsection{The full validation for the Lorenz system}
Combining all the parts of the previous sections, the algorithm proceeds as follows:
\begin{enumerate}
    \item The sum-of-squares program described in \cref{thm:sos} is formulated and solved approximately, with careful choices of basis for the right-hand side as listed in \cref{tab:bases}, leading to a polynomial Lyapunov function $V$ with coefficients $\mathbf{v}$ and symmetric Gram matrices $Q_e$, $R_e$, $Q_o$ and $R_o$ with coefficients $\mathbf{q}_e$, $\mathbf{r}_e$, $\mathbf{q}_o$ and $\mathbf{r}_o$ respectively.
    \item A large rectangular linear system $A\mathbf{y} = \mathbf{c}$ is constructed as in \cref{thm:linsos}, where $A$ and $\mathbf{c}$ are known exactly, and with approximate solution $\mathbf{y}_0 = \left(\mathbf{v}^T\, \mathbf{q}_e^T\, \mathbf{r}_e^T\, \mathbf{q}_o^T\, \mathbf{r}_o^T\right)^T$ given by the result of the SOS program.
    \item Degenerate rows are removed from this linear system to produce an equivalent system with full row rank $\hat{A}\mathbf{y} = \hat{\mathbf{c}}$.
    \item An interval matrix $[\hat{A}^+]$ enclosing the Moore-Penrose right-inverse $\hat{A}^+$ is calculated. This step also verifies that $\hat{A}$ has full row-rank.
    \item An interval vector $[\mathbf{y}]$ enclosing the true solution of $\hat{A}\mathbf{y} = \hat{\mathbf{c}}$ is computed as $[\mathbf{y}] = [\hat{A}^+]\hat{\mathbf{c}} + (I-[\hat{A}^+]A)\mathbf{y}_0$.
    \item The elements of $[\mathbf{y}]$ corresponding to the Gram matrices of the sum-of-squares program are used to assemble symmetric interval matrices $[Q_e]$, $[R_e]$, $[Q_o]$ and $[R_o]$.
    \item The determinant intervals of the principal minors of these Gram interval matrices are calculated. If they are all strictly positive, the matrices are positive-definite and the validation is successful.
\end{enumerate}

The bottleneck during our computations were the calculations using large interval matrices with very high precision - steps 4 and 5.
The fully validated results are listed in \cref{tab:results}.
At least with the 8th degree polynomials used, it was necessary to subdivide the ranges of $\rho$ into progressively smaller intervals as the values increased.
Combining these results, we have proved:

\begin{theorem}
\label{thm:final}
The Lorenz system \labelcref{eq:lorenz} with $\sigma=10$ and $\beta=8/3$ is globally stable for all $\rho\in[0,12]$.
\end{theorem}

\begin{table}[h]
    \centering
    \begin{tabular}{c|c|c|c}
         $\rho_{min}$ & $\rho_{max}$ & $d_x$ & $d_r$ \\
         \hline
         0 & 2 & 4 & 1 \\
         0 & 4 & 6 & 1 \\
         0 & 6 & 8 & 1 \\
         6 & 10 & 8 & 1 \\
         10 & 11 & 8 & 1 \\
         11 & 11.5 & 8 & 1 \\
         11.5 & 11.75 & 8 & 1 \\
         11.75 & 11.875 & 8 & 1 \\
         11.875 & 12 & 8 & 1
    \end{tabular}
    \caption{Successfully validated globally stable regions for $\rho\in[\rho_{min},\rho_{max}]$, with the corresponding degrees used for the Lyapunov function $V$.}
    \label{tab:results}
\end{table}

\section{Conclusion}
\label{sec:conc}

This work presents three new results. Firstly, a new Lyapunov function method has been described which can be used to prove a system is gradient-like, and in particular the absence of chaos, in continuous time dynamical systems. Secondly, this method has been applied to the Lorenz system up to $\rho=12$, using the inexact computational method of sum-of-squares optimisation, greatly extending the known parameter values for which the system is non-chaotic. Thirdly, the result has been rigorously validated using a new procedure based on interval arithmetic.


It is interesting to note that both the fixed-$\rho$ numerical results and the fully validated results failed to demonstrate gradient-like behaviour above $\rho=12$, even with significantly higher degree polynomials. There are several possible reasons for this. One is that this method would be possible in that range, with more carefully controlled numerics and higher-precision floating point numbers. The second is that the behaviour there is gradient-like, but that it is not possible to use our method to prove it, being only a sufficient condition. The third possibility is that there are some as-yet undiscovered periodic orbits or other non-trivial behaviour which exists above $\rho=12$ but below the known periodic orbits arising in the homoclinic explosion at $\rho\approx13.93$.

The validation procedure described in section \ref{sec:validation} is essentially proving a sharp bound of 0 for the method of \citet{fantuzzi2016bounds}. This could straightforwardly be extended to prove sharp bounds which are rational numbers, if care is taken to formulate the problem so that the bound is exactly specified as a floating-point number. However, it must be stressed that this method took a great deal of trial-and-error to implement, in order to choose the correct basis of polynomials, and so further development in this area is to be encouraged.

One potential extension of this work is to partial differential equations. Previous authors \citep{goluskin2019bounds,fuentes2022global} have applied sum-of-squares optimisation methods to dissipative PDEs using a suitably chosen function space to construct auxiliary functions. Initial experimentation suggests that the present method is equally applicable in such high-dimensional systems, though transforming the numerical results into rigorous ones becomes significantly more complicated.

\section*{Acknowledgements}
The author is indebted to David Goluskin who provided encouragement and advice on this work. Thanks are also due to Spyridon Pougkakiotis for useful comments. Olivier Hénot gave helpful suggestions for the Julia implementation.
The preliminary phase of this project was supported by the European Research Council (ERC) under the European Union's Horizon 2020 research and innovation programme (grant no. 865677), while the author was a member of the Emergent Complexity in Physical Systems laboratory at the École Polytechnique Fédérale de Lausanne.

The full source code to produce the figures in this paper and to perform the validation required for \cref{thm:final} is available at \url{https://github.com/jeremypparker/lorenzstability}.

\appendix
\section{An auxiliary function method for maps}
Consider a discrete-time dynamical system, that is a continuous map $f:\mathbb{R}^n\to\mathbb{R}^n$, and define a trajectory $(x_n)$ by
\begin{equation}
    x_{n+1} = f(x_n).
\end{equation}

\begin{definition}
The \textit{$\omega$-limit set} of a point $x$ with respect to a map $f$ is the subset of state space which is visited at arbitrarily high iterations,
\begin{equation}
    \omega(x,f) = \{y\in\mathbb{R}^n \mid \exists (n_k)\in\mathbb{N}\text{ with }n_{k+1}>n_k\text{ such that }f^{n_k}(x)\to y\text{ as }k\to\infty\}.
\end{equation}

A \textit{limit point} of the system is $y\in\mathbb{R}^n$ such that $y\in\omega(x,f)$ for some $x$.
\end{definition}

\begin{lemma}
Let $g:\mathbb{R}^n\to\mathbb{R}_{\geq 0}$ be some non-negative function defined continuously over state space. 
Suppose there exists a continuous function $V:\mathbb{R}^n\to\mathbb{R}$ such that, for all $x$,
\begin{equation}
    V(f(x)) - V(x) \geq g(x).
\end{equation}
Then $g(y)=0$ for all limit points $y$.
\end{lemma}
\begin{proof}
Let $y$ be a limit point and $(x_n)$ be a trajectory with a sequence of iterations $(n_k)$ such that $x_{n_k}\to y$.
Since $g(x)\geq 0$, we have that $V(f(x)) \geq V(x)$.
Then
\begin{align*}
    V(x_{n_{k+1}}) - V(x_{n_k})&= V\left(f^{n_{k+1}-n_k}(x_{n_k})\right) - V(x_{n_k}) \\
    &\geq V\left(f(x_{n_k})\right) - V(x_{n_k}) \\
    &\geq g(x_{n_k}) \geq 0.
\end{align*}
As $k\to\infty$, $V(x_{n_{k+1}}) - V(x_{n_k})\to 0$, and $g(x_{n_k})\to g(y)$.
So $0\geq g(y) \geq 0$, and thus $g(y)=0$.
\end{proof}

Now taking $g(x)=\left\|f(x)-x\right\|^2$, if such a $V$ exists we could use this result to show that all limit points are fixed points, and no other periodic points can exist. Furthermore, we could, for example, show that periodic points of period 1 (i.e. equilibria) or period $2$ can exist, but no higher, with $g(x) = \left\|f^2(x)-x\right\|^2$.
Any such result would show that the system is non-chaotic. In general let $g_k(x) = \left\|f^k(x)-x\right\|^2$.

A sum-of-squares relaxation is again possible for polynomial $f$.
We would need to find a $V$ such that
\begin{equation}
    V(f(x))-V(x)-g(x) \in \Sigma_n.
\end{equation}
Due to the composition of $V$ with $f$, the polynomial degrees rapidly become very large as the degree of $V$ is increased, which means that approximate solution and rigorous validation is both computationally expensive and requires high precision. Nevertheless, we can demonstrate the method in the simple case of the Hénon map, defined by
\begin{equation}
    f(x_1,x_2) = \left(1-ax_1^2+x_2,bx_1\right).
\end{equation}
We fix $b=0.3$, the standard value, and vary $a$. As $a$ increases from $0$ to the classical value of $1.4$, the initial fixed point undergoes a cascade of period-doubling bifurcations until the onset of chaos around $a\approx1.1$.

\Cref{fig:henon} shows the results of an implementation of this method with different possible degrees $d$ of the polynomial auxiliary function $V$ and different numbers $k$ of iterations  of the map. Even for $k=3$ the problem becomes very difficult and slow, but for $k=1$ we see a clear increase in the residual reported by the semi-definite solver (MOSEK) at the location of the first bifurcation, and for $k=2$ we see a clear increase at the second bifurcation when $d$ is sufficiently large.
The SDP solver reports nonzero residuals, so an exact solution has not been found and these results can only be taken as numerical evidence that it would be possible to find a suitable auxiliary function. Nevertheless, we have every reason to believe that a similar approach to that in section \ref{sec:validation} would be possible here, to rigorously validate these results.

\begin{figure}[h]
    \centering
    \includegraphics[width=0.8\textwidth]{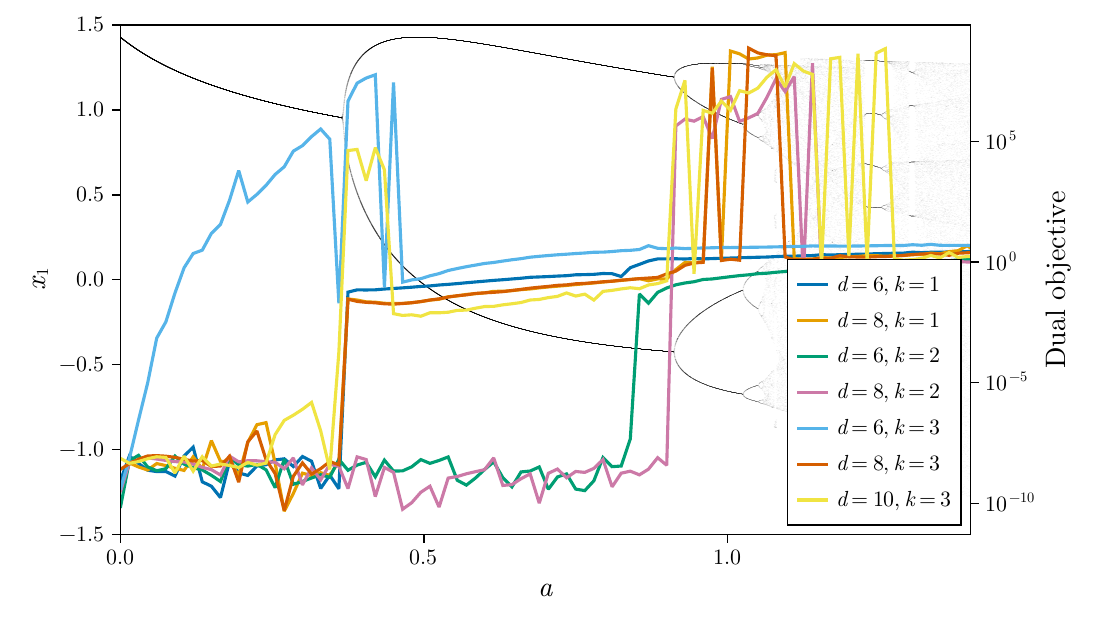}
    \caption{Bifurcation diagram (black, left axis) and SOS results (colours) for the Hénon map at $b=0.3$. For each choice of polynomial degree $d$ and map iterate $k$, the right vertical axis shows the dual objective value reported by MOSEK. A value close to zero indicates that the problem is likely to be feasible.}
    \label{fig:henon}
\end{figure}

\bibliography{bib}

\end{document}